\documentclass[letterpaper, 10 pt, conference]{ieeeconf}  

\IEEEoverridecommandlockouts                              
\makeatletter
\let\proof\@undefined                        
\let\endproof\@undefined                  
\makeatother 

\usepackage{cite}
\usepackage{amsmath,amssymb,amsfonts}
\usepackage{algorithmic}
\usepackage{graphicx}
\usepackage{textcomp}
\usepackage{amsmath,bm}

\usepackage{subfig}

\usepackage[font=singlespacing]{caption}
\usepackage{enumerate}

\usepackage{epsfig}
\usepackage{amsmath} 
\usepackage{amssymb}  
\usepackage{amsthm}
\usepackage{multirow}
\usepackage{bbm, dsfont}
\usepackage{bm} 
\usepackage{url}
\usepackage{balance}
\usepackage{color}
\graphicspath{{./figures/}}
\usepackage{array}
\usepackage[a]{esvect}

\usepackage{bbm}

\newtheorem{theorem}{Theorem}
\newtheorem{definition}{Definition}
\newtheorem{assumption}{Assumption}

\newtheorem{corollary}{Corollary}
\newtheorem{remark}{Remark}
\newtheorem{proposition}{Proposition}

\newtheorem{problem}{Problem}
\newtheorem{pro}{Problem 1.\hspace*{-0.2cm}}   

\usepackage[linesnumbered,ruled,vlined,onelanguage]{algorithm2e}

\definecolor{officegreen}{rgb}{0.1, 0.75, 0.0}

\newcommand{\yongn}[1]{{\color{black} #1}}
\newcommand{\yong}[1]{{\color{black} #1}}

\newcommand{\moh}[1]{{\color{black} #1}}
\newcommand{\jin}[1]{{\color{black} #1}}
\newcommand{\revise}[1]{{\color{black} #1}}
\newcommand{\moha}[1]{{\color{black} #1}}

\title{\LARGE \bf
\moh{Data-Driven Model Invalidation for Unknown Lipschitz Continuous Systems via Abstraction} 
}

\author{\moh{Zeyuan Jin\,$^{\rm \star}$, Mohammad Khajenejad\,$^{\rm \star}$ and Sze Zheng Yong}
\thanks{$^{\rm \star}$ These authors contributed equally to this paper.}
\thanks{The authors are with the School for Engineering of Matter, Transport and Energy,
Arizona State University, Tempe, AZ,  USA
       (email: {\tt \{zjin43,mkhajene,szyong\}@asu.edu})}%
       \thanks{This work was supported in part by DARPA grant D18AP00073.} 
}

\begin{document}

\maketitle

\begin{abstract}
In this paper, we consider the data-driven model invalidation problem for Lipschitz continuous systems, where instead of given mathematical models, only prior noisy sampled data of the systems are available. We show that this data-driven model invalidation problem can be solved using a tractable feasibility check. Our proposed approach consists of two main components: (i) a data-driven abstraction part that uses the noisy sampled data to over-approximate the unknown Lipschitz continuous dynamics with upper and lower functions, and (ii) an optimization-based model invalidation component that determines the incompatibility of the data-driven abstraction with a newly observed length-$T$ output trajectory. Finally, we discuss several methods to reduce the computational complexity of the algorithm and demonstrate their effectiveness with a simulation example of swarm intent identification. 
\end{abstract}


\section{Introduction}

\emph{Motivation.} In most Cyber-Physical Systems (CPS) applications, their analysis and design often require/assume the availability of mathematical models of the considered processes. 
Moreover, due to increasingly interconnected and integrated dynamics of such nonlinear, uncertain or hybrid systems, abstraction approaches have been developed to approximate the original complex dynamics with simpler dynamics \cite{Tabuada2009}. 
However, the precise model of the complex dynamics is often unknown, hence it is a challenging and interesting problem to determine ways to find abstractions and to analyze systems from only noisy sampled data. 

\emph{Literature Review.} \revise{\moha{The problem of determining whether an admissible model set \cite{smith1992model} can generate a finite sequence of experimental input-output data, \yong{known as model (in)validation, is useful for many control applications, including fault diagnosis and model identification \cite{harirchi2017guaranteed,Jin2019CDC}}.}} 
Several approaches for model invalidation have been recently developed for linear parameter varying systems \cite{bianchi2010robust, sznaier2003lmi}, nonlinear systems \cite{Prajna2006}, switched auto-regressive models \cite{ozay2014convex}, switched affine systems \cite{harirchi2017guaranteed,harirchi2018automatica}, etc., when their mathematical models are given.
Similarly, when mathematical models are available, \emph{abstraction} approaches have been widely studied for linear systems \cite{Girard2007}, nonlinear systems \cite{Girard2012,Singh2018CDC}, uncertain affine and nonlinear systems \cite{Shen2019ACC,Jin2019CDC}, 
and discrete-time hybrid systems \cite{Alimguzhin2017} in order to find simpler dynamics/systems that share most properties of interest with the original system dynamics for the sake of reducing computational complexity. However, these approaches are not applicable when accurate mathematical models are unavailable.


On the other hand, data-driven approaches 
that use sampled/observed input-output data to \emph{abstract} or over-approximate 
unknown dynamics using a bounded-error setting, where set-valued uncertainties are considered, 
have gained increased popularity over the last few years. 
The general objective of such data-driven methods is to find \emph{a set of known systems} that share the most properties of interest with the unknown system dynamics \cite{Milanese2004SetMI,canale2014nonlinear}. Under the assumption that the unknown dynamics is Lipschitz continuous, 
\cite{zabinsky2003optimal} provides a recursive algorithm to approximate upper and lower bounding functions 
for univariate functions, while \cite{beliakov2006interpolation} introduced  a novel computational approach for multivariate functions. The research in \cite{calliess2014conservative} further extended this approach to unknown dynamics that are H\"{o}lder continuous. Nonetheless, to our knowledge, these approaches do not explicitly deal with noise and their effect on the abstraction.

\emph{Contributions.} The goal of this paper is to tackle the problem of data-driven model invalidation by drawing upon 
model invalidation methods designed for when  mathematical models are available 
and data-driven approaches for finding abstractions/over-approximations of unknown dynamic systems from sampled data. Specifically, we propose a data-driven model invalidation algorithm which consists of two novel parts: (i) a data-driven abstraction component that over-approximates the unknown Lipschitz continuous dynamics from noisy sampled data, and (ii) an optimization-based model invalidation constituent that determines if the data-driven abstraction is incompatible with a new observed length-$T$ output trajectory. 
We further show that our data-driven model algorithm can be cast as a tractable feasibility check problem. In addition, we discuss and compare the use of several heuristic \revise{downsampling methods} for reducing the computational complexity of the algorithm, using an illustrative example of identifying swarm intent models.   

\section{Background} \vspace{-0.05cm}
\emph{Notation.} 
$\mathbb{R}^n$ denotes the $n$-dimensional Euclidean space \yong{and $\mathbb{R}^+$ is the set of all positive real numbers. 
For vectors $v,w \in \mathbb{R}^n$, $\|v\|_p \triangleq {\big (}\sum _{i=1}^{n}\left|x_{i}\right|^{p}{\big )}^{1/p}$, $1\leq p \leq \infty$ (in particular, $\| v \|_{\infty} \triangleq \max \limits_{1 \leq i \leq n} v_i$) and $v\leq w$ is a component-wise inequality. 
$\mathds{1}_m$ 
is an $m$-dimensional vector of ones.} 
\vspace{-0.05cm}
\subsection{Modeling Framework} \vspace{-0.05cm}
Consider a noisy discrete-time nonlinear auto-regressive  dynamic system model $\mathcal{G}$, at time step $k \geq 0$:
\begin{align}
\hspace*{-0.15cm}
y_{k+1} &= f(s_k) + w_k, \label{eq:system}\\
\tilde{y}_{k} &= {y}_{k} + v_k, \label{eq:output}
\hspace*{-0.15cm}
\end{align}
where $s_k\triangleq[y^{(1)}_k, \cdots, y^{(m)}_k,\cdots, y^{(1)}_{k-n_y+1}, \cdots, y^{(m)}_{k-n_y+1}]^\top \in \mathbb{R}^n$, $y_k \triangleq \begin{bmatrix}y^{(1)}_k, & y^{(2)}_k, & \cdots, & y^{(m)}_k\end{bmatrix}^\top \in \mathbb{R}^m$ and $f(\cdot) \triangleq \begin{bmatrix}f^1(\cdot) & \dots & f^i(\cdot) & \dots & f^m(\cdot) \end{bmatrix}^\top$ with 
 $f^i(\cdot):\mathbb{R}^n \rightarrow \mathbb{R}$ for all $ \ i \in \{1\dots m \}$ and $n=m n_y$, as well as process and measurement noise signals $w_k \in \mathcal{W}, v_k\in \mathcal{V}$ that are bounded, i.e., $\mathcal{W} \triangleq \{w_k \mid |w_k^{(i)}|\leq\varepsilon_w^{(i)},\forall i \in \{1,2,\cdots, m\}\}, \mathcal{V} \triangleq \{v_k\mid |v_k^{(i)}|\leq\varepsilon_v^{(i)},\forall i \in \{1,2,\cdots, m\}\}, $ with $\varepsilon_w^{(i)},\varepsilon_v^{(i)}>0$. 
 Functions $f^i(\cdot)$ are unknown but 
a noisy sampled data set $\mathcal{D}=\bigcup_{\ell=1}^N \mathcal{D}_\ell$ is available, consisting of $N$ trajectories each of length $T_\ell$ represented by $\mathcal{D}_\ell=\{\yongn{\tilde{y}_{j,\ell}}|j=0,\cdots, T_\ell-1\}$, 
where $\yongn{\tilde{y}_{k,\ell}}$ are noise corrupted measurements of $y_{k,\ell} \in \mathcal{Y}=[\mathcal{Y}_l, \mathcal{Y}_u]\subset \mathbb{R}^m$ according to \eqref{eq:output}. 

Moreover, we define $\yongn{\tilde{s}_{k,\ell}}\triangleq [(\yongn{\tilde{y}_{k,\ell}})^\top, \hdots, (\yongn{\tilde{y}_{k-n_y+1,\ell}})^\top]^\top $ and $\varepsilon_s$ as the upper bound of $\|s_{k,\ell}-\yongn{\tilde{s}_{k,\ell}}\|_p$ with $s_{k,\ell}\in\mathcal{S}$, \yongn{i.e., $\varepsilon_s = (n_y \sum_{i=1}^m (\varepsilon^{(i)}_v)^p)^\frac{1}{p}$ ($=\max_{i\in\{1,\dots,m\}} \varepsilon^{(i)}_v$ if $p=\infty$), 
and $\mathcal{S}$ is the $n_y$-ary Cartesian product of $\mathcal{Y}$}. 
For convenience, we also define a concatenated data set $\overline{\mathcal{D}}_\ell \triangleq \{(\yongn{\tilde{s}_{j,\ell}},\yongn{\tilde{y}_{j+1,\ell}})|j=n_y,\cdots, T_\ell-1\}$ that can \revise{be} constructed from $\mathcal{D}_\ell$ and similarly, $\overline{\mathcal{D}}=\bigcup_{\ell=1}^N \overline{\mathcal{D}}_\ell$.



Further, we assume continuity for 
$f^i(\cdot)$ \yongn{as follows:} 

\begin{assumption} \label{assumption:lip}
Each unknown vector field $f^i(\cdot)$, $\forall i \in \{1\dots m\}$, is $L^{(i)}_p$-Lipschitz continuous, i.e., there exists \yongn{a positive finite-valued} $L^{(i)}_p>0$, called \yongn{the} Lipschitz constant, such that $\forall x_1,x_2$ in domain of $f$, $|f^i(x_2)-f^i(x_1)| \leq L^{(i)}_p\|x_2-x_1\|_p$.
\end{assumption}
\vspace{-0.2cm}
\subsection{Abstraction/Over-Approximation}
\vspace{-0.1cm}
The goal of an abstraction procedure is to over-approximate the original (possibly unknown) function $f(\cdot) \moha{: \mathcal{S} \subset \mathbb{R}^n \to \mathbb{R}^m}$ by a pair of functions $\underline{f}$ and $\overline{f}$ \yongn{(i.e., to find an abstraction model $\mathcal{H}\triangleq\{\overline{f},\underline{f}\}$)} such that 
the function $f(\cdot)$ is bounded/sandwiched by the pair of functions, \yong{i.e., $\underline{f}$ and $\overline{f}$ satisfy} the following:
\begin{align} \label{eq1}
 \underline{f}(s) \leq f(s) \leq \overline{f}(s), \ \yongn{\forall s \in \yong{\mathcal{S}}}. 
\end{align}
\vspace{-0.6cm}
\subsection{Length-$T$ Behavior}
\vspace{-0.1cm}
Next, in preparation for the model invalidation problem, we adopt the definition in \cite{harirchi2017guaranteed} of the length-$T$ behavior of the original unknown model $\mathcal{G}$ and the abstracted model $\mathcal{H}$ based on the prior sampled data $\mathcal{D}$: 
\begin{definition}[Length-$T$ Behavior of Original Model $\mathcal{G}$] \label{behavior}
The length-$T$ behavior of the original (unknown) model $\mathcal{G}$ is the set of all length-$T$ output trajectories compatible with $\mathcal{G}$, given by the set 
\begin{align}
\hspace{-0.3cm}\begin{array}{rl}
 \mathcal{B}^T(\mathcal{G}):=&\{ \{\yongn{\tilde{y}_k}\}_{k=0}^{T-1}\mid \exists y_k\in\mathcal{Y}, w_k \in \mathcal{W}, v_k \in \mathcal{V}, \\
 &\ \text{ for } k \in \mathbb{Z}^{0}_{T-1}, \ \text{ s.t. \eqref{eq:system}--\eqref{eq:output} hold}
\}.
\end{array}
\end{align}
\end{definition}

\begin{definition}[Length-$T$ Behavior of Abstracted Model $\mathcal{H}$] \label{behavior}
The length-$T$ behavior of the abstracted model $\mathcal{H}$ is the set of all length-$T$ output trajectories compatible with $\mathcal{H}$, given by the set 
\begin{align}
\hspace{-0.3cm}\begin{array}{rl}
 \mathcal{B}^T(\mathcal{H}):=&\{ \{\yongn{\tilde{y}_k}\}_{k=0}^{T-1}\mid \exists y_k\in\mathcal{Y}, w_k \in \mathcal{W}, v_k \in \mathcal{V},\\ 
 &\ \text{ for } k \in \mathbb{Z}^{0}_{T-1}, \ \text{ s.t. \eqref{eq:output}--\eqref{eq1} hold}\}.
\end{array}
\end{align}
\end{definition}
Using the above definitions of system behaviors as well as the fact that $\mathcal{H}$ is an abstraction of $\mathcal{G}$ (by construction), we can conclude that $\mathcal{B}^T(\mathcal{G}) \subseteq \mathcal{B}^T(\mathcal{H})$.

\section {Problem Statement}
We now state the data-driven model invalidation problem that we consider in this paper: 
\begin{problem}[Model Invalidation for $\mathcal{G}$] \label{problem1}
Given a new sequence of output trajectory $\{\tilde{y}^n_k\}_{k=0}^{T-1}$, an unknown target model $\mathcal{G}$, for which only prior sampled data $\mathcal{D}_\mathcal{G}$ is available, and an integer $T$, determine whether the trajectory belongs to the target model, i.e., to determine if the following holds: 
\begin{align}\label{eq:prob2}
\hspace*{-0.15cm}\begin{array}{l}
\{\tilde{y}^n_k\}_{k=0}^{T-1}\in \mathcal{B}^T(\mathcal{G}).
\end{array}\hspace*{-0.15cm}
\end{align}
\end{problem}
However,  \yongn{it is non-trivial to solve Problem \ref{problem1}} directly because only prior data set $\mathcal{D}_\mathcal{G}$ from the original model $\mathcal{G}$ with unknown dynamics is available. Hence, we aim to solve a stricter auxiliary problem 
that, if solved, also provides a solution to Problem \ref{problem1}. A two-step process is taken, where the first step is solving the following problem to obtain an abstraction model of the unknown dynamics $\mathcal{G}$: 
\begin{pro}[Data-Driven Abstraction] 
\label{problem1.1}
For a set of $N$ sampling data points $\mathcal{D}_\mathcal{G}$, 
find a pair of upper and lower functions $\overline{f}$ and $\underline{f}$ \yongn{(i.e., $\mathcal{H}\triangleq \{\overline{f},\underline{f}\}$)} such that:
\begin{align}
\hspace*{-0.15cm}
\underline{f}(s)\leq f(s) \leq \overline{f}(s), \ \yongn{\forall s \in \mathcal{S},}
\hspace*{-0.15cm}
\end{align}
where $f(\cdot) \yongn{: \mathcal{S} \subset \mathbb{R}^n \to \mathbb{R}^m}$ is the original unknown dynamics of the system, and correspondingly determine   $\mathcal{B}^T(\mathcal{H})$.
\end{pro}

The second step is to solve the following model invalidation problem for the abstracted models.
\begin{pro}[Model Invalidation for $\mathcal{H}$]  \label{problem1.2}
Given a new sequence of output trajectory $\{\tilde{y}^n_k\}_{k=0}^{T-1}$, an abstraction model $\mathcal{H}$ of target model $\mathcal{G}$ and an integer $T$, determine whether the trajectory belongs to the target model. That is, to determine if the following is true,
\begin{align}\label{eq:prob2}
\hspace*{-0.15cm}\begin{array}{l}
\{\tilde{y}^n_k\}_{k=0}^{T-1}\in \mathcal{B}^T(\mathcal{H}).
\end{array}\hspace*{-0.15cm}
\end{align}
\end{pro}
\yong{Next, we show that the solution to Problem 1.\ref{problem1.1} and 1.\ref{problem1.2} is \emph{sufficient} to solve Problem \ref{problem1}.}
\yong{Note that this sufficient (only) condition renders it inapplicable for model validation (due to interpolation errors and noise), but is useful for the model invalidation problem that we consider in this paper.}
\begin{proposition} \label{propMI}
Suppose that $\mathcal{B}^T(\mathcal{G}) \subseteq \mathcal{B}^T(\mathcal{H})$. Then, $\{\tilde{y}^n_k\}_{k=0}^{T-1}\notin \mathcal{B}^T(\mathcal{G})$ if $\{\tilde{y}^n_k\}_{k=0}^{T-1}\notin \mathcal{B}^T(\mathcal{H})$ and $\{\tilde{y}^n_k\}_{k=0}^{T-1}\in \mathcal{B}^T(\mathcal{H})$ if $\{\tilde{y}^n_k\}_{k=0}^{T-1}\in \mathcal{B}^T(\mathcal{G})$.
\end{proposition}
\begin{proof}
If the output trajectory $\{\tilde{y}^n_k\}_{k=0}^{T-1}$ is excluded from model $\mathcal{H}$, i.e., $\{\tilde{y}^n_k\}_{k=0}^{T-1} \cap \mathcal{B}^T(\mathcal{H})=\emptyset$, it is also excluded from model $\mathcal{G}$ since $\mathcal{B}^T(\mathcal{G}) \subseteq \mathcal{B}^T(\mathcal{H})=\emptyset$. On the other hand, when $\{\tilde{y}^n_k\}_{k=0}^{T-1}\cap  \mathcal{B}^T(\mathcal{G})\neq\emptyset$, then necessarily $\{\tilde{y}^n_k\}_{k=0}^{T-1}\cap\mathcal{B}^T(\mathcal{H})\neq \emptyset$ since $\mathcal{B}^T(\mathcal{H}) \supseteq \mathcal{B}^T(\mathcal{G})$. 
\end{proof}

\section {Data-driven Abstraction and Model Invalidation}
In this section, we first introduce an approach for obtaining a data-driven abstraction of the unknown system \eqref{eq:system}, before proposing an optimization-based approach to \yongn{invalidate} the resulting abstraction model with given new noisy output trajectories. The two algorithms we propose solve Problems 1.\ref{problem1.1} and 1.\ref{problem1.2}, and consequently, Problem \ref{problem1} by Proposition \ref{propMI}.
\subsection {Data-Driven Abstraction Algorithm} 
\begin{theorem} \label{LipschitzInterpolation}{}
Consider system \eqref{eq:system} and its corresponding data set $\overline{\mathcal{D}}=\bigcup_{\ell=1}^N\{(\yongn{\tilde{s}_{j,\ell}}, \yongn{\tilde{y}_{j+1,\ell}})|\moh{j}=n_y, \cdots, T_\ell-1\}$. \moh{Suppose Assumption \ref{assumption:lip} holds. Then, for all $s\in \mathcal{S}$, $\underline{f}(\cdot)$ and $\overline{f}(\cdot)$ are lower and upper abstraction functions for unknown function $f(\cdot)$, i.e., $\forall s \in \mathcal{S}$, $\underline{f}(s) \leq f(s) \leq \overline{f}(s)$,  where $\underline{f}(\cdot) \triangleq \begin{bmatrix}\underline{f}^1(\cdot) & \dots & \underline{f}^i(\cdot) & \dots & \underline{f}^m(\cdot) \end{bmatrix}^\top$, $\overline{f}(\cdot) \triangleq \begin{bmatrix}\overline{f}^1(\cdot) & \dots & \overline{f}^i(\cdot) & \dots & \overline{f}^m(\cdot) \end{bmatrix}^\top$ and} 
 \begin{subequations}
 \begin{align} 
& \overline{f}\moh{^i}(s) \hspace{-0.05cm}= \hspace{-0.25cm} \min_{\substack{j \in \{n_y,\hdots, T_\ell-1\},\\ \ell \in \{1,\hdots,N\}}} \hspace{-0.1cm} ((\yongn{\tilde{y}_{j+1,\ell}})^{(i)}\hspace{-0.1cm}+\hspace{-0.1cm}L^{(i)}_p\|s\hspace{-0.1cm}-\hspace{-0.1cm}\yongn{\tilde{s}_{j,\ell}}\|_p)\hspace{-0.1cm}+\hspace{-0.1cm}\varepsilon^{(i)}_{t} \hspace{-0.05cm}, \label{upper_func}\\[-0.2em]
 & \underline{f}\moh{^i}(s)  \hspace{-0.05cm}= \hspace{-0.25cm}\max_{\substack{j \in \{n_y,\hdots, T_\ell-1\},\\ \ell \in \{1,\hdots,N\}}}\hspace{-0.1cm} ((\yongn{\tilde{y}_{j+1,\ell}})^{(i)}\hspace{-0.1cm}-\hspace{-0.1cm}L^{(i)}_p\|s\hspace{-0.1cm}-\hspace{-0.1cm}\yongn{\tilde{s}_{j,\ell}}\|_p) \hspace{-0.1cm}-\hspace{-0.1cm}\varepsilon^{(i)}_{t} \hspace{-0.05cm}, \label{lower_func}
 \end{align}
 with $\varepsilon^{(i)}_t \triangleq \varepsilon^{(i)}_w+\varepsilon^{(i)}_v+L^{(i)}_p\varepsilon_s$, \moh{for all $i \in \{1,\dots, m \}$.  Moreover, $\underline{f}^i$ and $\overline{f}^i$ are $L^{(i)}_p$-Lipschitz continuous functions and $\underline{f}$ and $\overline{f}$ are also Lipschitz continuous.}
 \end{subequations}
\end{theorem}
\begin{proof}
It follows from Assumption \ref{assumption:lip} that
$$|f^i(x_2)-f^i(x_1)| \leq L^{(i)}_p\|x_2-x_1\|_p, \forall i \in \{1\dots m\}.$$
Let $x_2=s$ be the generic variable and $x_1=\yongn{s_{j,\ell}}=\yongn{\tilde{s}_{j,\ell}}-\yongn{z_{j,\ell}}$ be the de-noised observation of the noisy $\yongn{\tilde{s}_{j,\ell}}$, where $\yongn{z_{j,\ell}}$ satisfies $\|\yongn{z_{j,\ell}}\|_p \leq \varepsilon_s$. Furthermore, from \eqref{eq:system} and \eqref{eq:output}, we have $f^i(\yongn{s_{j,\ell}})=(\yongn{y_{j+1,\ell}})^{(i)}-(\yongn{w_{j,\ell}})^{(i)}=(\yongn{\tilde{y}_{j+1,\ell}})^{(i)}-(\yongn{w_{j,\ell}})^{(i)}-(\yongn{v_{j,\ell}})^{(i)}$.

Next, applying the triangle inequality to 
\hspace{-0.3cm}\begin{gather*}(\yongn{\tilde{y}_{j+1,\ell}})^{(i)}-(\yongn{w_{j,\ell}})^{(i)}-(\yongn{v_{j,\ell}})^{(i)}-L^{(i)}_p\|s-\yongn{\tilde{s}_{j,\ell}}-\yongn{z_{j,\ell}}\|_p \leq \\ \hspace{-0.01cm} f^i(s) \hspace{-0.07cm}\leq \hspace{-0.07cm}(\yongn{\tilde{y}_{j+1,\ell}})^{(i)} \hspace{-0.08cm}-\hspace{-0.08cm}(\yongn{w_{j,\ell}})^{(i)}\hspace{-0.08cm}-\hspace{-0.08cm}(\yongn{v_{j,\ell}})^{(i)}\hspace{-0.08cm}+\hspace{-0.08cm}L^{(i)}_p\hspace{-0.06cm}\|s\hspace{-0.08cm}-\hspace{-0.08cm}\yongn{\tilde{s}_{j,\ell}}\hspace{-0.08cm}-\hspace{-0.08cm}\yongn{z_{j,\ell}}\|_p,\end{gather*}
the results in \eqref{upper_func} and \eqref{lower_func} follow from the fact that these inequalities should hold for \emph{all} the sampled data and \emph{all} the possible values of noise signals. $L^{(i)}_p$-Lipschitz continuity of upper and lower abstraction functions is implied by the fact that function $\|.\|_p$ is $1$-Lipschitz continuous and \cite{sukharev1978optimal}.  
\end{proof}
\begin{corollary} \label{linear}
 \moh{If $1$ or $\infty$} norm \moh{is considered, i.e., $p=1$ or $\infty$, then} the abstraction \moh{functions are} piecewise affine function\moh{s}, and the corresponding abstraction is called a piecewise affine abstraction.
\end{corollary}
\begin{proposition} \label{Abs}
\moh{The abstraction approach described in Theorem \ref{LipschitzInterpolation} satisfies \emph{monotonicity}, in the sense that 
given two data sets $\mathcal{D}$ and $\mathcal{D}'$}, $\mathcal{D}' \subseteq \mathcal{D}$ \moh{implies that} the abstraction model $\mathcal{H}_{\mathcal{D}'}$ over-approximates the abstraction model $\mathcal{H}_\mathcal{D}$.
\end{proposition}
\begin{proof} \moh{Let $J$ and $J'$ be the set of indices corresponding to data pairs included in $\overline{\mathcal{D}}$ and $\overline{\mathcal{D}'}$ (constructed from $\mathcal{D}$ and $\mathcal{D}'$) and $\forall i \in \{1\dots m\}$, $\overline{f}^i(s)$ and  $\overline{f'}^i(s)$, and $\underline{f}^i(s)$ and $\underline{f'}^i(s)$ are upper and lower  abstraction functions returned by the abstraction models $\mathcal{H}_{\mathcal{D}}$ and $\mathcal{H}_{\mathcal{D}'}$, respectively.} Then, $\mathcal{D}' \subseteq \mathcal{D} \implies \overline{\mathcal{D}'} \subseteq \overline{\mathcal{D}} \implies J' \subseteq J \implies \overline{f}^i(s)=$ $ \min_{j,\ell \in J} ((\yongn{\tilde{y}_{j+1, \ell}})^{(i)}+L^{(i)}_p\|s-\yongn{\tilde{s}_{j,\ell}}\|_p+\varepsilon^{(i)}_{t}) \leq  \min_{j,\ell \in J'} ((\yongn{\tilde{y}_{j+1,\ell}})^{(i)}+L^{(i)}_p\|s-\yongn{\tilde{s}_{j,\ell}}\|_p+\varepsilon^{(i)}_{t})= \overline{f'}^i(s) \implies $ $\overline{f}(s) \leq \overline{f}'(s)$ \moh{, where the second inequality holds since the two optimization problems have the same objective functions, but the constraint set of the latter is a subset of the former. By a similar argument,} $\underline{f}(s) \geq \underline{f}'(s)$. \moh{It follows from these two results that} $\mathcal{H}_{\mathcal{D}'}$ over-approximates $\mathcal{H}_\mathcal{D}$.
\end{proof}
\subsection {Data-Driven Model Invalidation Algorithm}
 
 We apply an optimization-based model invalidation approach to determine if the data-driven abstraction obtained in the previous section is incompatible with new observed length-$T$ output trajectory. 
In particular, we propose \yongn{a} model invalidation algorithm 
\yongn{for} abstraction model $\mathcal{H}$ \yongn{as follows:}

 
\begin{theorem} \label{theorem2}
Given an abstracted model $\mathcal{H}$, 
\yongn{a} \revise{new} \yongn{observed} length-$T$ output sequence \moha{$\{\tilde{y}^n_k\}_{k=0}^{k=T-1}$} \revise{\yong{invalidates} model $\mathcal{H}$, if the following feasibility problem is infeasible:}
 \begin{subequations}
 \begin{align}
 \nonumber & \hspace{-0.3cm} \text{Find} \ {y_k}, w_k, v_k \ \forall  k  \in\mathbb{Z}_{T-1}^0         \\[-0.em]
 \nonumber &\hspace{-0.3cm} \text{subject to} \  \forall  k  \in\mathbb{Z}_{T-1}^{n_y}, \forall  \ell  \in\mathbb{Z}_{N}^{1}, \forall (\yongn{\tilde{s}_{j,\ell}},\yongn{\tilde{y}_{j+1,\ell}})\in \overline{\mathcal{D}}_\ell: \\
 & \hspace{0.3cm} y_{k+1} \leq \yongn{\tilde{y}_{j+1, \ell}}+L_p||s_k-\yongn{\tilde{s}_{j, \ell}}||_p +\varepsilon_t+w_k, \label{constraint1'}\\[-0.em] 
 & \hspace{0.3cm} y_{k+1} \geq \yongn{\tilde{y}_{j+1, \ell}}-L_p||s_k-\yongn{\tilde{s}_{j, \ell}}||_p -\varepsilon_t+w_k, \label{constraint2'}\\[-0.em] 
   & \hspace{0.3cm} \forall k \in \mathbb{Z}_{T-1}^0:  
  \tilde{y}^n_k = y_k+v_k, \ \underline{y}_k \leq y_k \leq \overline{y}_k,\\ 
    & \hspace{0.3cm}-\varepsilon_w \mathds{1}_m\leq w_k\leq \varepsilon_w \mathds{1}_m, -\varepsilon_v \mathds{1}_m\leq v_k\leq \varepsilon_v \mathds{1}_m,
 \end{align}
 \end{subequations}
 where $\overline{\mathcal{D}_\ell} = \{(\yongn{\tilde{s}_{j,\ell}}, \yongn{\tilde{y}_{j+1, \ell}})|j=n_y, \cdots, T_\ell-1\}$ is one trajectory of $\overline{\mathcal{D}}$ and  $\overline{\mathcal{D}}=\bigcup_{\ell=1}^N \overline{\mathcal{D}}_\ell$ is the given sampled data set from which we obtain a data-driven abstraction $\mathcal{H}$ with $\underline{y}_k$ and $\overline{y}_k$ as given bounds of $y_k$, $s_k=[y_k,\cdots,y_{k-n_y+1}]^T$, 
 $\varepsilon_t=[\varepsilon_t^{(1)}, \hdots, \varepsilon_t^{(m)}]^\top$,  $L_p=[L_p^{(1)}, \hdots, L_p^{(m)}]^\top$ \yongn{and $\varepsilon_t^{(i)}$ for all $i \in \{1,\hdots,m\}$ is as defined in Theorem 1}.
\end{theorem}
\begin{proof}
By the definition of model invalidation for abstracted model $\mathcal{H}$, we know that 
the abstraction  is invalidated if the following problem is infeasible:
 \begin{subequations} \label{MI_o}
 \begin{align} \tag{MI}
 \nonumber & \hspace{-0.3cm} \text{Find} \ {y_k}, w_k, v_k\ \forall  k  \in\mathbb{Z}_{T-1}^0         \\[-0.em]
 &\hspace{-0.3cm} \text{subject to} \  \forall  k  \in\mathbb{Z}_{T-1}^{n_y}: 
  y_{k+1} \leq \overline{f}(s_k) +w_k,\label{constraint_1}\\[-0.em]
 & \hspace{3.15cm} y_{k+1} \geq \underline{f}(s_k) +w_k,\label{constraint_2}\\[-0.em]
  & \hspace{0.5cm} \forall k \in \mathbb{Z}_{T-1}^0:  
  \tilde{y}^n_k = y_k+v_k, \ \underline{y}_k \leq y_k \leq \overline{y}_k,\\ 
    & \hspace{0.5cm}-\varepsilon_w \mathds{1}_m\leq w_k\leq \varepsilon_w \mathds{1}_m, -\varepsilon_v \mathds{1}_m\leq v_k\leq \varepsilon_v \mathds{1}_m.
 \end{align}
 \end{subequations}
 
Since the upper bound of the abstraction is given by \eqref{upper_func}, constraint \eqref{constraint_1} is equivalent to \eqref{constraint1'}. 
Similarly, by \eqref{lower_func}, \eqref{constraint2'} is equivalent to \eqref{constraint_2}. Thus, two optimization problems are equivalent. If the above optimization problem is infeasible, it means that the output sequence $\{\tilde{y}^n_k\}_{k=0}^{k=T-1}$ cannot be consistent with the length-$T$ behavior of $\mathcal{H}$, i.e., $\{\tilde{y}^n_k\}_{k=0}^{k=T-1} \notin \mathcal{B}^T(\mathcal{H})$, hence the model is invalidated.
\end{proof}

From Corollary \ref{linear}, by choosing the suitable vector norm, i.e., $p=1$ or $\infty$, the optimization problem in Theorem \ref{theorem2} is a \jin{mixed integer} linear program/feasibility problem. 

\begin{remark}\label{rem:1} By Proposition \ref{Abs}, we could choose to only use a strict subset of the sampled data, $\overline{\mathcal{D}}_{\ell} (\yongn{\tilde{y}^n_k})\subset \overline{\mathcal{D}}_\ell$, without violating the guarantees of Theorem \ref{theorem2}, where $\overline{\mathcal{D}}_{\ell} (\yongn{\tilde{y}^n_k})$ may be a function of the new observed data $\yongn{\tilde{y}^n_k}$. In this case, we will replace \eqref{constraint1'} and \eqref{constraint2'} with the following constraints:
\begin{align*}
y_{k+1} \leq \yongn{\tilde{y}_{j+1,\ell,k}}+L_p||s_k-\yongn{\tilde{s}_{j,\ell,k}}||_p +\varepsilon_t+w_k,\\
y_{k+1} \geq \yongn{\tilde{y}_{j+1,\ell,k}}-L_p||s_k-\yongn{\tilde{s}_{j,\ell,k}}||_p -\varepsilon_t+w_k,
\end{align*}
for all $(\yongn{\tilde{s}_{j,\ell,k}},\yongn{\tilde{y}_{j+1,\ell,k}}) \in \overline{\mathcal{D}}_{\ell} (\yongn{\tilde{y}^n_k})$. 
The advantage of this ``downsampling" is that the computational time can be reduced but at the cost of the abstracted model precision and thus the ability of the data-driven model invalidation algorithm to eliminate inconsistent models. We will explore this downsampling strategy in the simulation section.
\end{remark}

\vspace{-0.3cm}
\subsection{Estimation of Lipschitz Constant} \vspace{-0.05cm}
In previous sections, the Lipschitz constants are assumed to be given. 
\yong{In the case when the constants are not known}, we will 
\yongn{estimate the Lipschitz constant from the noisy} sampled data set $\overline{\mathcal{D}}=\{(\yongn{\tilde{s}_j}, \yongn{\tilde{y}_{j+1}})|j=n_y, \cdots, N-1\}$ as follows: 
\yongn{\begin{align} \label{estimateL}
\hat{L}^{(i)}_p = \max\big\{0,\max_{j\neq k} \frac{|(\yongn{\tilde{y}_{j+1}})^{(i)} -(\yongn{\tilde{y}_{k+1}})^{(i)}|-2\varepsilon_v}{||\yongn{\tilde{s}_j}-\yongn{\tilde{s}_k}||_p + 2 \varepsilon_s}\big\}.
\end{align}
This is an extension of the lazy approach in \cite[Section 4.3.2]{calliess2014conservative} to the case where both the input and output data, i.e., $\tilde{s}_j$ and $\tilde{y}_{j+1}$ for all $j$, are corrupted by bounded noise. The above expression can be simply obtained from the definition of Lipschitz continuity and the use of triangle inequality.}

Since the accuracy of $L^{(i)}_p$ is crucial for the results in the previous section, we proceed to find some guarantees that we obtain the right estimate with high probability. To achieve this, we leverage a classical result on probably approximately correct (PAC) learning for linear separators, which is summarized below: \vspace{-0.1cm}
\begin{definition}[Linear Separators\cite{kearns1994introduction}] \label{def:linear} For $\Gamma \subset \mathbb{R} \times \mathbb{R}$, a linear separator is a pair $(a, b)\in \mathbb{R}^2$ such that
\begin{align} \label{eq:linsep}
\forall(x, y) \in \Gamma: \ x \leq ay +b.
\end{align}
\end{definition}\vspace{-0.25cm}

\begin{proposition}[PAC Learning \cite{kearns1994introduction}] 
Let $\epsilon, \delta \in \mathbb{R}^+$. If number of sampling points is $N\geq\frac{1}{\epsilon}\ln\frac{1}{\delta}$, where the sample points $\Gamma$  are drawn from a distribution $\mathcal{P}$, then, with probability greater than $1 - \delta$, a linear separator $(a,b)$ has an error $\texttt{err}_{\mathcal{P}}$ of less than $\epsilon$, where the error of a pair $(a, b)$ is defined as $\texttt{err}_{\mathcal{P}}(a,b)=\mathcal{P}({(x,y)\in\Gamma|x>ay+b})$.
\end{proposition}\vspace{-0.05cm}

From Definition \ref{def:linear}, it is easy to verify that our \yongn{$\hat{L}_P^{(i)}$ estimate}  \moha{in} \eqref{estimateL} is a special case of the linear separator \yongn{in} \eqref{eq:linsep} with $b=0$. Thus, the estimated $\hat{L}_P^{(i)}$ using \eqref{estimateL} is guaranteed \moha{to be close to the true Lipschitz constant of the original unknown function} with high probability if we have sufficient data. 
\vspace{-0.05cm}
\section{Simulation and Discussion} \vspace{-0.05cm}
In this section, we demonstrate the effectiveness of the proposed methods for data-driven abstraction and model invalidation for swarm intent/formation identification. Since the computational complexity for abstraction and model invalidation will be high when dealing with large data sets, we also consider downsampling as described in Remark \ref{rem:1}. The downsampling strategies we propose involve taking only \yongn{a local subset of the points in the dataset into consideration} 
instead of all \yongn{data}, 
and thereby, the computational time could be reduced. 
In particular, we will focus on grid\yongn{-based}, $k$-means and $k$-nearest neighbors (kNN) methods for picking \yongn{the local subset} 
and compare their performances. All simulations are implemented in MATLAB on a 2.2 GHz machine with 16 GB of memory. Yalmip \cite{YALMIP} and Gurobi \cite{gurobi} are used for implementation 
of both data-driven abstraction and model invalidation algorithms. 

\vspace{-0.15cm}
\subsection{System Dynamics and Data Set Generation} \label{Dynamics}
\vspace{-0.1cm}
In this section, we describe the dynamics of the swarm intent/formation models and 
\yongn{expound} the data set generation process that we employ for the simulations.

\subsubsection{System Dynamics} 
The dynamics of each swarm agent is described by the Dubins Car model \cite{Dubins1957}: 
\begin{subequations} \label{eq:dynamics}
\begin{align}
&p_{x,k+1}=p_{x,k}+u_s\cos(\theta_k)\delta t +w_{px,k}, \label{func_x}\\ 
& p_{y,{k+1}}=p_{y,k}+u_s\sin(\theta_k)\delta t +w_{py,k}, \label{func_y}\\ 
& \theta_{k+1} = \theta_k +\frac{u_s}{L}\tan(u_\phi)\delta t +w_{\theta,k}, \label{func_theta}
\end{align}
\end{subequations}
where the system states are $p_x$ and $p_y$ that represent the $(x,y)$-position of the agent and $\theta$ as the heading angle of the agent. 
$L$ is the length between the front and rear tires and is set to $1.5m$, $u_s$ is the speed of the agent and is assumed to be $1\, m/s$, 
and sampling time $\delta t$ is set to $0.1 s$,  while $w_{px,k}$, $w_{py,k}$ and $w_{\theta,k}$ represent process noise or heterogeneity among the agents and are set to be bounded by $|w_{px,k}| \le 0.00025$, $|w_{py,k}| \le 0.00025$ and $|w_{\theta,k}| \le 0.0001$, respectively. 
In addition, a reference signal $\theta_{desired,k}$ based on the centroid of the swarm formation $(c_x, c_y)$ is assumed to be given: 
\begin{align}
\theta_{desired,k}=\arctan2(c_y-p_{y,k}, c_x-p_{x,k}),
\end{align}
which the agents utilize for feedback control according to 
the following proportional control law:
\begin{align}\label{controller}
u_\phi = \min(\frac{\pi}{8}, \max(-\frac{\pi}{8}, K_p(\theta_{desired}-\theta))),
\end{align}
where the saturation functions ensure that the steering angle of each agent never exceeds $[-\frac{\pi}{8}, \frac{\pi}{8}] \ rad$. 

\subsubsection{Data Set Generation} 

We consider two swarm intents or formations, which are dependent on the choice of the value of $K_p$. When $K_p = 0.5$, the swarm \emph{intends} to move towards the centroid of the swarm and in our simulation examples, the sampled data set $\mathcal{D}$ (with size $|\mathcal{D}|$) is collected using this model (see Fig. \ref{fig:D}). On the other hand, when $K_p = -0.5$, the swarm \emph{intends} to move \moha{away from} the centroid. In our simulation examples, we use this model to generate the new observed trajectories $\{\tilde{y}^n_k\}_{k=0}^{T-1}$ (see Fig. \ref{fig:y}) to invalidate the abstraction based on the sampled data set $\mathcal{D}$. 
\begin{figure}[t] 
\centering \vspace{-0.1cm}\hspace{-0.1cm}
\subfloat[Model for generating $\mathcal{D}$ \yongn{($K_p$~\,\,\, $=0.5$)} \label{fig:D}]{\includegraphics[width=0.2375\textwidth,trim=1mm 3mm 2mm 6.mm,clip]{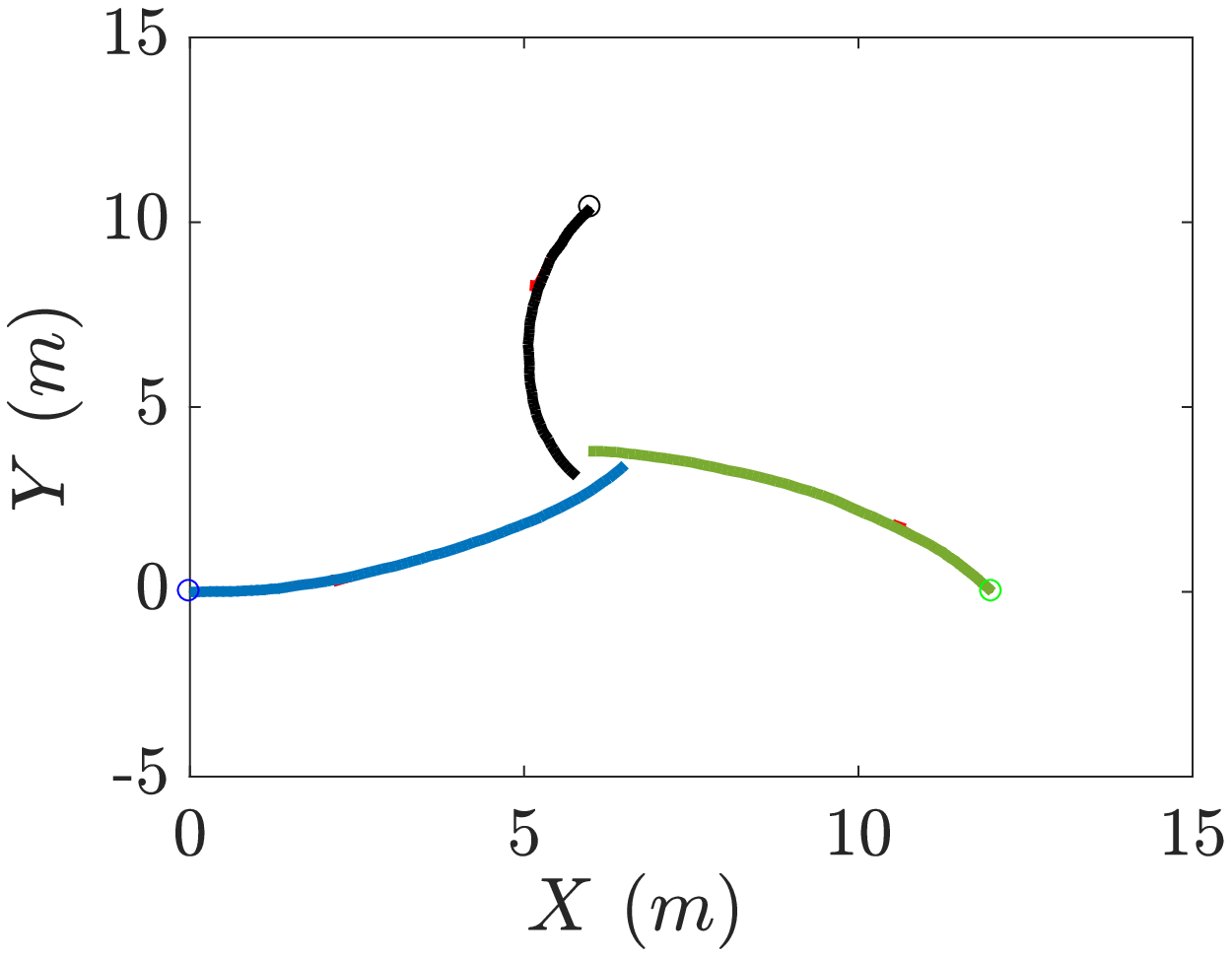}\hspace{-0.1cm}}
\subfloat[Model for generating new observed trajectory $\{\tilde{y}^n_k\}_{k=0}^{T-1}$ \yongn{($K_p=-0.5$)} \label{fig:y}]{~\includegraphics[width=0.2375\textwidth,trim=2mm 2mm 2mm 7.mm,clip]{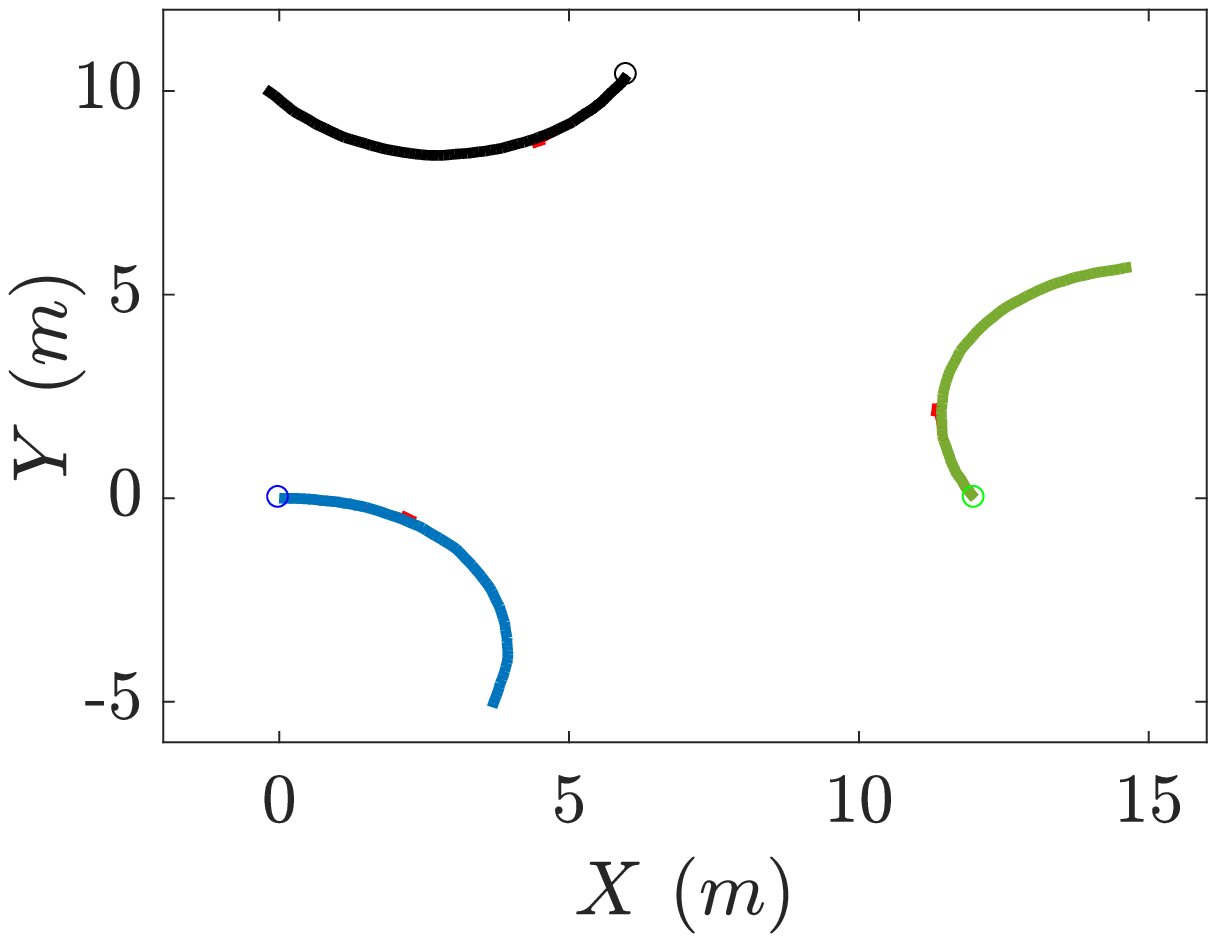}~}%

\caption{Illustration of two swarm formation/intention models.}%
\label{figure1}\vspace{-0.5cm}
\end{figure}

In both data sets, we assume that we have 3 agents and noisy observations of  $y_k = [p_{x,k}, p_{y,k}, \theta_k]^T$ are available with the following initial conditions for the 3 agents: $[0, 0, 0]$, $[12, 0, \frac{2\pi}{3}]$ and $[6, 6\sqrt{3}, -\frac{2\pi}{3}]$ representing $p_{x,0}$, $p_{y,0}$ and $\theta_0$, respectively. Moreover, we chose $T=16$ and generated 20 random sequences \moha{of} trajectories for the simulations.
 

\vspace{-0.2cm}
\subsection{Simulation Results}
\vspace{-0.1cm}
\subsubsection{Data-Driven Abstraction} First, we applied the data-driven abstraction algorithm to the nonlinear functions in \eqref{eq:dynamics}, i.e., the $u_s \cos(\theta_k)$ and $u_s \sin(\theta_k)$ terms.
For illustrative purposes, in Figure \ref{figure1}, we show the data-driven abstraction results of the function $f(u, \theta) = u\cos(\theta)$ defined on the domain $u\times \theta\in[-2, 2]\times[0,2\pi]$, where we additionally assume that $u$ is measured with a noise bound of 0.1, and similarly, the bounds of the noise of measuring $\theta$ and $f(u, \theta)$ are also assumed to be 0.1. As expected, the resulting abstraction is indeed an over-approximation of the unknown nonlinear function on the defined domains of interest. Moreover, as we increase the number of sampled data point\yongn{s} $\mathcal{D}$, the more accurate the over-approximation is.



\subsubsection{Model Invalidation Results}
Next, using the data-driven abstractions of the unknown dynamics of \eqref{eq:dynamics} based on the sampled data $\mathcal{D}$, we compare the performances of the model invalidation algorithm when using all the sample data (i.e., no downsampling) as we vary the size of the data set \yong{by selecting and only using a subset of the data set $\mathcal{D}$ for each time step as described in Remark \ref{rem:1}. Specifically,} 
we consider three \yong{\emph{heuristic}} downsampling approaches \yong{and compare their performances and associated CPU times.}
\paragraph{Without Downsampling}
First, we consider the case without downsampling, i.e., the entire data set $\mathcal{D}$ is used for model invalidation. Specifically, we vary the size of the data set, $|\mathcal{D}|$, from 16 to 208, and compare the ability of the data-driven model invalidation algorithm to invalidate the (wrong) model with $K_p=0.5$ using 20 randomly generated noisy new observed trajectories from the true model with $K_p=-0.5$. As shown in Figure \ref{fig:all}, when the data set size is small, the (wrong) model is never invalidated by any of the 20 new trajectories, but when the data set size is increased, the model is invalidated by more new trajectory data.

\begin{figure}[t] 
\centering
\includegraphics[width=0.325\textwidth,trim=0mm 2mm 0mm 0mm,clip]{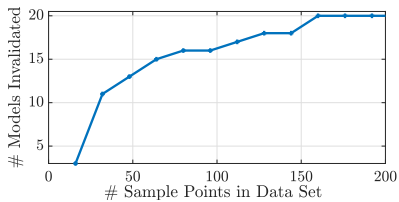}

\caption{Performance of \yongn{the} data-driven modeling invalidation algorithm as the prior sampled data size is increased. \label{fig:all}}%
\label{figure1}\vspace{-0.25cm}
\end{figure}

\paragraph{With Downsampling using the Grid\yongn{-Based} Method}
The first downsampling strategy we consider uses a grid\moha{-based} method to select a subset of  dataset $\mathcal{D}$ (see Remark \ref{rem:1} for implementation details). The idea is to uniformly grid the 
\yongn{domain} into \yongn{hyperrectangles} and assign sampled data to each region. We then \moha{pick} the region based on the new data $\tilde{y}^n_k$ and use only the sampled data assigned to this region in the constraints described in Remark \ref{rem:1}. 
Moreover, since the states near the boundaries of the region may be poorly 
approximated, we 
further consider 
\yong{extra} sample\moha{d} data from neighboring regions by randomly adding a few more 
sampled data points. 
This downsampling method \yongn{has the benefit} \moha{that} the associated region for each new data $\tilde{y}^n_k$ can be easily found ($\mathcal{O}(1)$), \yongn{but scalability is an issue since} 
the number of regions grows exponentially with the state 
dimension.

Figure \ref{fig:grid} depicts the performance of this downsampling strategy from a data set of size 160 in terms of the number of invalidated models as well as its mean CPU time over 20 new trajectories. As the grid size increases, the size of the downsampled data set decreases, resulting in worse invalidation performance but the mean CPU time does initially decrease before increasing again. This may be 
\yongn{due to a} 
 trade-off between decreased space complexity of the sample points \moha{and} increased space complexity of the grid points. 

\vspace{-0.cm}
\paragraph{With Downsampling using the $k$-Means Method}
Next, we consider \moha{a}  
clustering\moha{-based downsampling approach,} known as \yongn{the} $k$-means algorithm  \cite{bishop:2006:PRML}. \moha{It} is a well-known unsupervised learning algorithm that groups similar data into clusters. Using this, we partition the original data set $\mathcal{D}$ into several clusters and only use the sampled data associated with the \moha{closest} cluster \moha{to} the new data $\tilde{y}^n_k$ in the constraints described in Remark \ref{rem:1}. 
Further, \moha{to address the potentially poor over-approximation near the boundaries of the cluster}, one additional random data point \moha{is picked} from each \moha{of the} other cluster\moha{s}.  
The scalability of this method is better than the grid\moha{-based} method, since the number of regions depends only on the number of clusters. \moha{However, computing the clusters might be} computationally \moha{expensive}. 

Figure \ref{fig:cluster} depicts the ability of this downsampling strategy 
to invalidate the model as well as its mean CPU time over 20 new trajectories. As the cluster size decreases, the size of the downsampled data set decreases, thus the  invalidation performance \moha{becomes} worse, but the CPU time improves.

\begin{figure}[t] 
\centering
\includegraphics[width=0.52\textwidth,trim=20mm 2mm 9mm 0.5mm,clip]{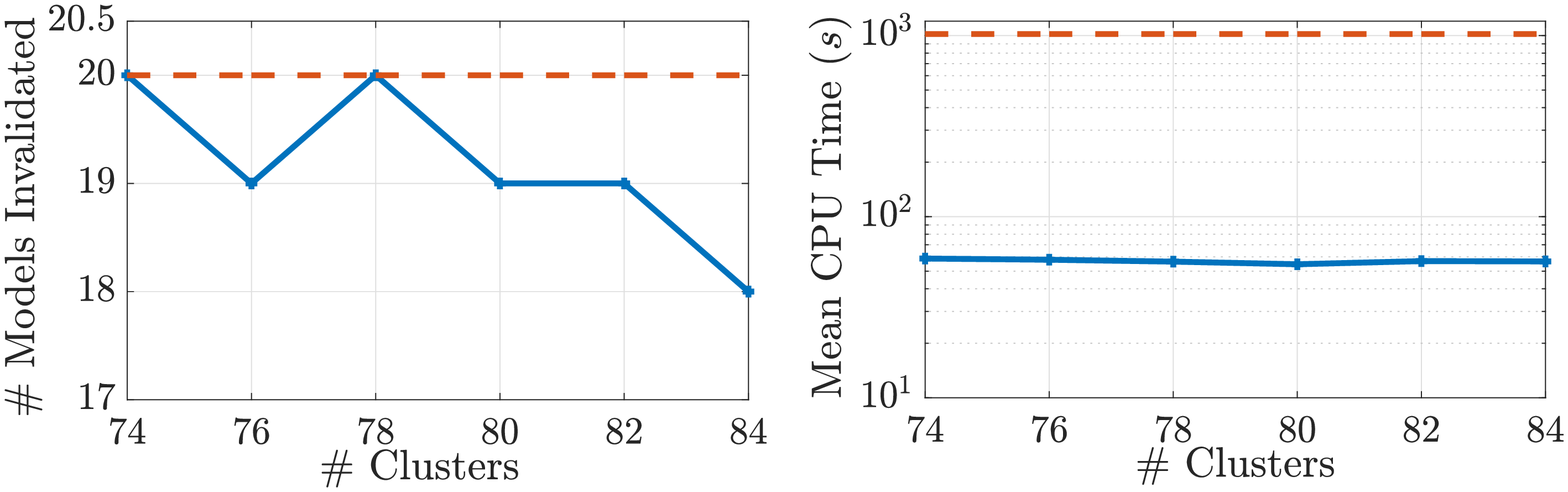}

\caption{Performance \yongn{with} 
$k$-Means based downsampling. The blue solid lines represent the performance of the $k$-Means based method and the red dash lines are the reference performance when using the entire data\yongn{set} without downsampling.}%
\label{fig:cluster}\vspace{-0.475cm}
\end{figure}


\vspace{-0.01cm}
\paragraph{With Downsampling using the $k$-Nearest Neighbors (kNN) Method}
Finally, we consider an \revise{kNN method for downsampling, inspired by \cite[Section 4.2.6]{calliess2014conservative}}, where we only use the $k$ closest data points to the new data $\tilde{y}^n_k$ in the constraints described in Remark \ref{rem:1}. 
This strategy is a form of \emph{lazy learning} with no required preprocessing nor extra memory.  
Since only $k$ points are taken into consideration for each $\tilde{y}^n_k$,  computational time and number of constraints for the resulting algorithm can both decrease.  
However, the complexity of kNN algorithm, $\mathcal{O}(nd+kn)$, indicates that it may be time-consuming to formulate the 
model invalidation problem for a large data set \yongn{with kNN-based downsampling}. 

\moha{The results in} Figure \ref{fig:knn} \moha{indicate that} the \moha{kNN-based} downsampling strategy \moha{significantly} improves \moha{the} mean CPU time, \yongn{when} compared to \moha{the \yongn{reference} performance}. Moreover, as the number of 
nearest neighbors increase\moha{s}, the performance of the model invalidation algorithm improves at the cost of a slight increase in the mean CPU time. 

\yong{In conclusion, we observed that the heuristic downsampling techniques significantly reduce the CPU times but deteriorate the model invalidation performance. A principled analysis of the trade-off between CPU times and  performance will be a subject of future research.}
\begin{figure}[t] 
\centering
\includegraphics[width=0.52\textwidth,trim=20mm 2mm 9mm 0.5mm,clip]{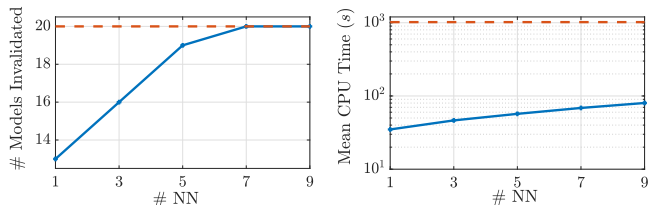}

\caption{Performance \yongn{with} 
kNN-based downsampling. The blue solid lines represent the performance of the kNN-based method and the red dash lines are the reference performance when using the entire data\yongn{set} without downsampling.}%
\label{fig:knn}\vspace{-0.2cm}
\end{figure}

\vspace{-0.1cm}
\section{Conclusion}
We proposed a data-driven approach for model invalidation for unknown Lipschitz continuous systems where only noisy sampled data is available. 
 In the first step, we introduced an algorithm to find upper and lower Lipschitz functions that over-approximate/abstract our unknown original Lipschitz continuous dynamics from noisy data. Then, we proposed a data-driven model Invalidation algorithm for 
 determining the (in)compatibility of the abstracted data-driven model with a new observed length-$T$ output trajectory, 
 that we showed is equivalent to a tractable linear feasibility program. 
Finally, we applied our proposed approach to an example of swarm intent identification and compared several downsampling strategies to reduce the  
computational complexity of the proposed algorithm.  

\bibliographystyle{unsrt}
\bibliography{biblio}

\end{document}